\documentclass[pra,aps,showpacs,twocolumn,twoside,superscriptaddress]{revtex4-1}

\usepackage{amsmath,amsfonts,amssymb,caption,color,epsfig,graphics,graphicx,hyperref,latexsym,mathrsfs,revsymb,theorem,url,verbatim,epstopdf}

\usepackage{amssymb}
\usepackage{color}
\usepackage{amsmath,epic,curves,amscd}
\usepackage[english]{babel}
\usepackage{graphicx}
\usepackage{comment}
\usepackage{appendix}
\usepackage{mathdots}

\newtheorem{definition}{Definition}
\newtheorem{proposition}[definition]{Proposition}
\newtheorem{lemma}[definition]{Lemma}

\newtheorem{theorem}[definition]{Theorem}
\newtheorem{corollary}[definition]{Corollary}
\newtheorem{conjecture}[definition]{Conjecture}

\newtheorem{remark}[definition]{Remark}
\newtheorem{example}[definition]{Example}
\newtheorem{question}[definition]{Question}

\def\squareforqed{\hbox{\rlap{$\sqcap$}$\sqcup$}}
\def\qed{\ifmmode\squareforqed\else{\unskip\nobreak\hfil
\penalty50\hskip1em\null\nobreak\hfil\squareforqed
\parfillskip=0pt\finalhyphendemerits=0\endgraf}\fi}
\def\endenv{\ifmmode\;\else{\unskip\nobreak\hfil
\penalty50\hskip1em\null\nobreak\hfil\;
\parfillskip=0pt\finalhyphendemerits=0\endgraf}\fi}
\newenvironment{proof}{\noindent \textbf{{Proof.~} }}{\qed}
\def\Dbar{\leavevmode\lower.6ex\hbox to 0pt
{\hskip-.23ex\accent"16\hss}D}
\makeatletter
\def\url@leostyle{%
  \@ifundefined{selectfont}{\def\UrlFont{\sf}}{\def\UrlFont{\small\ttfamily}}}
\makeatother
\def\bcj{\begin{conjecture}}
\def\ecj{\end{conjecture}}
\def\bcr{\begin{corollary}}
\def\ecr{\end{corollary}}
\def\bd{\begin{definition}}
\def\ed{\end{definition}}
\def\bea{\begin{eqnarray}}
\def\eea{\end{eqnarray}}
\def\bem{\begin{enumerate}}
\def\eem{\end{enumerate}}
\def\bex{\begin{example}}
\def\eex{\end{example}}
\def\bim{\begin{itemize}}
\def\eim{\end{itemize}}
\def\bl{\begin{lemma}}
\def\el{\end{lemma}}
\def\bma{\begin{bmatrix}}
\def\ema{\end{bmatrix}}
\def\bpf{\begin{proof}}
\def\epf{\end{proof}}
\def\bpp{\begin{proposition}}
\def\epp{\end{proposition}}
\def\bqu{\begin{question}}
\def\equ{\end{question}}
\def\br{\begin{remark}}
\def\er{\end{remark}}
\def\bt{\begin{theorem}}
\def\et{\end{theorem}}

\def\btb{\begin{tabular}}
\def\etb{\end{tabular}}

\newcommand{\nc}{\newcommand}

\def\a{\alpha}
\def\b{\beta}
\def\g{\gamma}

\def\r{\rho}
\def\s{\sigma}

\def\ph{\varphi}

\def\ps{\psi}

\def\G{\Gamma}

\def\T{\Theta}

 \nc{\bbA}{\mathbb{A}} \nc{\bbB}{\mathbb{B}} \nc{\bbC}{\mathbb{C}}
 \nc{\bbD}{\mathbb{D}} \nc{\bbE}{\mathbb{E}} \nc{\bbF}{\mathbb{F}}
 \nc{\bbG}{\mathbb{G}} \nc{\bbH}{\mathbb{H}} \nc{\bbI}{\mathbb{I}}
 \nc{\bbJ}{\mathbb{J}} \nc{\bbK}{\mathbb{K}} \nc{\bbL}{\mathbb{L}}
 \nc{\bbM}{\mathbb{M}} \nc{\bbN}{\mathbb{N}} \nc{\bbO}{\mathbb{O}}
 \nc{\bbP}{\mathbb{P}} \nc{\bbQ}{\mathbb{Q}} \nc{\bbR}{\mathbb{R}}
 \nc{\bbS}{\mathbb{S}} \nc{\bbT}{\mathbb{T}} \nc{\bbU}{\mathbb{U}}
 \nc{\bbV}{\mathbb{V}} \nc{\bbW}{\mathbb{W}} \nc{\bbX}{\mathbb{X}}
 \nc{\bbZ}{\mathbb{Z}}

 \nc{\bA}{{\bf A}} \nc{\bB}{{\bf B}} \nc{\bC}{{\bf C}}
 \nc{\bD}{{\bf D}} \nc{\bE}{{\bf E}} \nc{\bF}{{\bf F}}
 \nc{\bG}{{\bf G}} \nc{\bH}{{\bf H}} \nc{\bI}{{\bf I}}
 \nc{\bJ}{{\bf J}} \nc{\bK}{{\bf K}} \nc{\bL}{{\bf L}}
 \nc{\bM}{{\bf M}} \nc{\bN}{{\bf N}} \nc{\bO}{{\bf O}}
 \nc{\bP}{{\bf P}} \nc{\bQ}{{\bf Q}} \nc{\bR}{{\bf R}}
 \nc{\bS}{{\bf S}} \nc{\bT}{{\bf T}} \nc{\bU}{{\bf U}}
 \nc{\bV}{{\bf V}} \nc{\bW}{{\bf W}} \nc{\bX}{{\bf X}}
 \nc{\bZ}{{\bf Z}}

\nc{\cA}{{\cal A}} \nc{\cB}{{\cal B}} \nc{\cC}{{\cal C}}
\nc{\cD}{{\cal D}} \nc{\cE}{{\cal E}} \nc{\cF}{{\cal F}}
\nc{\cG}{{\cal G}} \nc{\cH}{{\cal H}} \nc{\cI}{{\cal I}}
\nc{\cJ}{{\cal J}} \nc{\cK}{{\cal K}} \nc{\cL}{{\cal L}}
\nc{\cM}{{\cal M}} \nc{\cN}{{\cal N}} \nc{\cO}{{\cal O}}
\nc{\cP}{{\cal P}} \nc{\cQ}{{\cal Q}} \nc{\cR}{{\cal R}}
\nc{\cS}{{\cal S}} \nc{\cT}{{\cal T}} \nc{\cU}{{\cal U}}
\nc{\cV}{{\cal V}} \nc{\cW}{{\cal W}} \nc{\cX}{{\cal X}}
\nc{\cY}{{\cal Y}}
\nc{\cZ}{{\cal Z}}

\nc{\hA}{{\hat{A}}} \nc{\hB}{{\hat{B}}} \nc{\hC}{{\hat{C}}}
\nc{\hD}{{\hat{D}}} \nc{\hE}{{\hat{E}}} \nc{\hF}{{\hat{F}}}
\nc{\hG}{{\hat{G}}} \nc{\hH}{{\hat{H}}} \nc{\hI}{{\hat{I}}}
\nc{\hJ}{{\hat{J}}} \nc{\hK}{{\hat{K}}} \nc{\hL}{{\hat{L}}}
\nc{\hM}{{\hat{M}}} \nc{\hN}{{\hat{N}}} \nc{\hO}{{\hat{O}}}
\nc{\hP}{{\hat{P}}} \nc{\hR}{{\hat{R}}} \nc{\hS}{{\hat{S}}}
\nc{\hT}{{\hat{T}}} \nc{\hU}{{\hat{U}}} \nc{\hV}{{\hat{V}}}
\nc{\hW}{{\hat{W}}} \nc{\hX}{{\hat{X}}} \nc{\hZ}{{\hat{Z}}}

\nc{\hn}{{\hat{n}}}

\def\dim{\mathop{\rm Dim}}

\def\lin{\mathop{\rm span}}

\def\max{\mathop{\rm max}}
\def\min{\mathop{\rm min}}

\def\rank{\mathop{\rm rank}}

\def\tr{\mathop{\rm Tr}}

\def\dg{\dagger}

\def\op{\oplus}
\def\ox{\otimes}

\newcommand{\bra}[1]{\langle#1|}
\newcommand{\ket}[1]{|#1\rangle}
\newcommand{\proj}[1]{| #1\rangle\!\langle #1 |}

\newcommand{\braket}[2]{\langle#1|#2\rangle}

\newcommand{\tbc}{\red{TO BE CONTINUED...}}

\newcommand{\opp}{\red{OPEN PROBLEMS}.~}

\newcommand{\red}{\textcolor{red}}

\newcommand{\jmp}{J. Math. Phys.}

\def\Dbar{\leavevmode\lower.6ex\hbox to 0pt
{\hskip-.23ex\accent"16\hss}D}

\begin{document}

\title{Nonexistence of $n$-qubit unextendible product bases of size $2^n-5$}

\author{Lin Chen}
\email{linchen@buaa.edu.cn (corresponding author)}
\affiliation{School of Mathematics and Systems Science, Beihang University, Beijing 100191, China}
\affiliation{International Research Institute for Multidisciplinary Science, Beihang University, Beijing 100191, China}

\def\Dbar{\leavevmode\lower.6ex\hbox to 0pt
{\hskip-.23ex\accent"16\hss}D}
\author {{ Dragomir {\v{Z} \Dbar}okovi{\'c}}}
\email{djokovic@uwaterloo.ca}
\affiliation{Department of Pure Mathematics and Institute for
Quantum Computing, University of Waterloo, Waterloo, Ontario, N2L
3G1, Canada} 

\date{\today}


\begin{abstract}
It is known that the $n$-qubit system has no unextendible product bases (UPBs) of cardinality $2^n-1$, $2^n-2$ and $2^n-3$. On the other hand the $n$-qubit UPBs of cardinality $2^n-4$ exist for all $n\ge3$. We prove that they do not exist for cardinality $2^n-5$.
\end{abstract}

\maketitle


\section{Introduction}

The notion of UPBs is fundamental in quantum information theory and has various applications. First, UPBs have been constructed to characterize the nonlocality without entanglement that appears when locally distinguishing product vectors \cite{bdf99}. Second, UPBs have been used to construct positive-partial-transpose (PPT) entangled states \cite{bdm99}. All two-qutrit UPBs have been constructed as well as all two-qutrit PPT entangled states of rank four \cite{cd11}. Furthermore, the multiqubit UPBs have been used to construct Bell inequalities 
without quantum violation \cite{afk12}. Recently, it has been shown that the structure of multiqubit UPBs is related to the so-called orthogonality complete graphs, and many multiqubit  
UPBs have been thus constructed \cite{johnston13}. However the main problem, namely to determine the cardinalities of multiqubit UPBs is still unresolved.

For convenience we denote by $\T_n$ the set of cardinalities of UPBs in the $n$-qubit systems. As we allow a UPB to span the whole Hilbert space, we have $2^n\in\T_n$. It is known that $2^n-1$, $2^n-2$ and $2^n-3$ do not belong to $\T_n$ for any 
$n$, and $2^n-4\in\T_n$ for all $n\ge3$  \cite{johnston14,johnston14c}. 
In these references, extensive computer computations failed 
to find any example of $n$-qubit UPBs of cardinality $2^n-5$. 
In the cases $n=3,4$ it is known that they do not exist.
Hence, the question was raised whether such UPBs exist for some $n\ge5$. In Theorem \ref{thm:the} we prove that they do not exist. Our proof is based on the study of the hypothetical entangled PPT projector $\r$ of rank 5 which a UPB of size 
$2^n-5$ would provide. As a result we constructed a couple of  examples of 5-qubit bound entangled states of rank 5, which  give the affirmative answer to a question raised by Johnston in \cite[Sec. 6]{johnston14}, open problem (1).

The rest of this paper is organized as follows. In Sec. \ref{sec:rank5} we introduce our notation and recall some known facts. Then we prove two auxilliary lemmas and our main result in Sec. \ref{sec:main}. We conclude in Sec. \ref{sec:con}.

\section{Preliminaries} 
\label{sec:rank5}

Let $\cH=\cH_1\ox\cdots\ox\cH_n$ be the Hilbert space of dimension $D$ representing a quantum system $A_1,\ldots,A_n$ consisting of $n$ parties. We are mainly interested in the case where all parties are qubits, i.e. each Hilbert space $\cH_j$ has dimension two. In that case, we fix an orthonormal basis $\ket{0}_j,\ket{1}_j$ of $\cH_j$. Usually, the subscript $j$ will be suppressed. 

We say that a vector $\ket{v}\in\cH$ is a {\em unit vector} 
if $\|v\|=1$. As a rule, we shall not distinguish two unit vectors which differ only in the phase. When $\dim\cH_j=2$, by using this convention, we can say that for any unit vector 
$\ket{v_j}\in\cH_j$ there exists a unique unit vector 
$\ket{v_j^\perp}\in\cH_j$ perpendicular to $\ket{v_j}$. 

A {\em product vector} is a nonzero vector 
$\ket{x}=\ket{x_1}\ox\cdots\ox\ket{x_n}$, which will be 
written also as $\ket{x}=\ket{x_1,\ldots,x_n}$. 
If $\|x\|=1$ we shall assume (as we may) that 
each $\|x_j\|=1$. 
Two product vectors $\ket{x}=\ket{x_1,\ldots,x_n}$ and 
$\ket{y}=\ket{y_1,\ldots,y_n}$ are orthogonal if and only if 
$\ket{x_j}\perp\ket{y_j}$ for at least one index $j$. 
We use the abbreviation {\em OPS} to denote any set of pairwise 
orthogonal unit product vectors in $\cH$. 
The cardinality of an OPS cannot exceed $D$, the 
dimension of $\cH$. We say that an OPS is an OPB, 
{\em orthogonal product basis}, if its cardinality is $D$.
As an example, in the $n$-qubit system, the $2^n$ product vectors $\ket{x_s}=\ket{s_1,\ldots,s_n}$, where 
$s:=(s_1,\ldots,s_n)$ runs through all binary 
$\{0,1\}$-sequences of length $n$, is an OPB. 
We refer to this OPB as the {\em standard OPB}. 
However, there are many other $n$-qubit OPBs and describing or 
classifying them for any $n$ is a very hard problem, see 
\cite{bdm99,DiV03,bravyi,fs09,cd12,cd13,cj15,cd16,johnston13} and our paper \cite{cd16} for the case $n=4$. 

An {\em unextendible product basis} (UPB) is an OPS such that there is no product vector orthogonal to all vectors of the OPS  \cite{bdm99,DiV03}. Originally it was required that UPB does not span the whole Hilbert space $\cH$, but for the sake of convenience we have dropped that restriction. We say that a UPB is {\em proper} if it does not span $\cH$. We record some facts from the introduction section in the following lemma.
\bl
\label{le:maxcardinal}
For $n\ge3$, the two largest integers in $\Theta_n$ are $2^n$ and $2^n-4$.
\el

On the other hand, the problem of finding the smallest element of $\Theta_n$ has been considered by several authors 
\cite{AL01,KF06,bravyi} and it was finally resolved by Johnston \cite[Theorem 1]{johnston13}. We state his theorem as follows.

\bt
\label{thm:upb}
(Jonhston, 2013) The smallest integer in $\Theta_n$ is:

(i) $n+1$ if $n$ is odd;

(ii) $n+2$ if $n\equiv2 \pmod{4}$;

(iii) $n+4$ if $n\equiv0 \pmod{4}$ and $n>8$;

(iv) $6$ if $n=4$ and $11$ if $n=8$.
\et

For small values of $n$, we have $\Theta_1=\{2\}$, 
$\Theta_2=\{4\}$, $\Theta_3=\{4,8\}$. Further, from 
\cite[Table 3]{johnston14} we see that 
\bea
\T_4 &=& \{6,7,8,9,10,12,16\}, \\
\T_5 &\supseteq& \{6,8,9,10,12-26,28,32\}, \\
\T_6 &\supseteq& \{8,9,12,14-58,60,64\}, \\
\T_7 &\supseteq& \{8,12,16,17,18,20-122,124,128\},
\eea 
where $i-j$ means that all integers $k$ in the range 
$i\le k\le j$ are included. 
Johnston \cite{johnston14} asks whether the integers $11,27$ belong to $\T_5$; $10,11,13,59$ to $\T_6$; and $10,11,13,14,15,19,123$ to $\T_7$.

The main result of this note, Theorem \ref{thm:the}, asserts that 
$2^n-5\notin\T_n$ for any $n$. In particular, we obtain partial answers to the above questions, namely $27\notin\Theta_5$, 
$59\notin\Theta_6$, and $123\notin\Theta_7$. 

In what follows we recall some known results which we need for the proof of our main result. Here we allow the spaces $\cH_i$ to have any finite dimension $\ge 2$.

Let $\r$ denote an $n$-partite state $\r:=\r_{12\cdots n}$ on the space $\cH_1\ox\cH_2 \ox \cdots \ox\cH_n$. We say that $\r$ is a  {\em PPT state} if the partial transpose of $\r$ with respect to any subsystem is positive semidefinite. 
We use the acronym PPTES to denote entangled PPT states. 
If $\r$ is a separable state, then $L(\r)$ denotes its length. 
For any linear operator $\r$ we denote its range by $\cR(\r)$ and its nullspace by $\ker\r$.

We say that a bipartite state $\r$ is a $k\times l$ {\em state} if its reduced states $\r_1$ and $\r_2$ have ranks $k$ and $l$, respectively. 
In the bipartite case, $\r^\G$ will denote the partial transpose of $\r$ with respect to the first system. 
We refer to the ordered pair $(\rank\r,\rank\r^\G)$ as the 
{\em birank} of $\r$.

\begin{lemma}
\label{le:supposerho}
(i) Let $\rho$ be a $2\times N$ PPT state and 
$\ket{a,b}\in\ker\r\cap(\cR(\r_1)\ox\cR(\r_2))$ a unit product vector. Then there exists $\lambda>0$ such that 
$\s:=\r-\lambda\proj{a,b}\ge0$, $\rank\s=\rank\r-1$, 
$\rank\s^\G=\rank\r^\G-1$, $\rank\s_2=\rank\r_2-1=N-1$, and 
the state $\s$ is PPT. 

(ii) If $\r$ is a $2\times2$ or $2\times3$ separable state of birank $(r,s)$, then $L(\r)=\max(r,s)$.
\qed
\end{lemma} 
\begin{proof}
(i) For all assertions, except the PPT property of $\s$, see 
\cite[Lemma 7]{kck00} and its proof. We shall prove the PPT property. By tracing out the first party in the equation 
$\r=\s+\lambda\proj{a,b}$, we obtain that 
$\r_2=\s_2+\lambda\proj{b}$. It follows that 
$\ket{b}\notin\cR(\s_2)$. Hence, there exists an invertible 
linear operator $W$ on $\cH_2$ which makes $P:=W\s_2 W^\dg$ into  an orthogonal projector and such that $PW\ket{b}=0$. 
If $V:=W^{-1}PW$ then $(I\ox V)\r=(I\ox V)\s$ because 
$PW\ket{b}=0$, and so we have 
$(I\ox V)\r(I\ox V^\dg)=(I\ox V)\s(I\ox V^\dg)$. Thus to prove 
that $\s$ is PPT it suffices to verify that 
$(I\ox V)\s(I\ox V^\dg)=\sigma$. This can be done as follows.

Let $\s_2=\sum^{N-1}_{i=1}\proj{b_j}$ and $P=\sum^{N-1}_{i=1}\proj{j}$. Since $P=W\s_2 W^\dg$ we may assume that $W\ket{b_j}=\ket{j}$. Since $\s_2$ is a reduced density operator of $\s$, we have $\s=\sum_i \proj{\a_i}$ where $\ket{\a_i}=\sum^r_{j=1} \ket{a_{j,i},b_j}$ for some vectors $\ket{a_{j,i}}$. Then one can verify that $(I\ox V)\s(I\ox V^\dg)=\sigma$.

(ii) is proved in \cite[Proposition 3]{cd12}. 
\end{proof}

In the following lemma, $r_j$ denotes the rank of the $j$th reduced density operator, $\r_j$, of the state $\r$. 

\begin{lemma} {\rm (see \cite{cd13}).} \label{le:theppt}
(i) Any $n$-partite PPT state of rank at most three is separable. 

(ii) If $\r$ is an $n$-partite PPTES of rank four then either $n=2$ and $r_1=r_2=3$ or $n=3$ and $r_1=r_2=r_3=2$. 
\qed
\end{lemma}

The following lemma is a special case of Kruskal's theorem 
(see \cite{JBK1977} and 
\cite[Theorem 12.5.3.1, p. 306]{JML2012}).
\bl
\label{le:Kruskal}
Let $\ket{a}=\ket{a_1,a_2,a_3}$, $\ket{b}=\ket{b_1,b_2,b_3}$, 
$\ket{c}=\ket{c_1,c_2,c_3}$ and $\ket{d}=\ket{d_1,d_2,d_3}$  
be product vectors of a tripartite system and let the vectors 
$\ket{a_i}$ and $\ket{b_i}$ be linearly independent for 
$i=1,2,3$. Then the equality $\ket{a}+\ket{b}=\ket{c}+\ket{d}$ 
implies that, up to phase factors, $\{\ket{a},\ket{b}\}=\{\ket{c},\ket{d}\}$. 
\el

For the convenience of the reader, let us state
\cite[Remark, p. 7]{cd16} as a lemma.

\bl \label{le:2xN}
All OPBs of a bipartite system $\cH_1\ox\cH_2$ with 
$\dim\cH_1=2$ can be constructed by the following method. First choose an orthogonal decomposition 
$\cH_2=X_1\oplus \cdots \oplus X_m$ with $k_j:=\dim X_j \ge 1$
and $\sum k_j=\dim\cH_2$. 
Next choose $m$ pairwise different o.n. bases 
$\{\ket{v_j},\ket{v_j^\perp}\}$, $j\in\{1,2,\ldots,m\}$ of 
$\cH_1$. 
Finally, for each $j$, choose two arbitrary o.n. bases 
$\{ \ket{x_{j,1}},\ket{x_{j,2}},\ldots,\ket{x_{j,k_j}}\}$ and 
$\{ \ket{y_{j,1}},\ket{y_{j,2}},\ldots,\ket{y_{j,k_j}}\}$ of $X_j$. 
Then the product vectors 
\begin{eqnarray*}
&&
\ket{v_j,x_{j,1}},\ldots,\ket{v_j,x_{j,k_j}}, \\
&&
\ket{v_j^\perp,y_{j,1}},\ldots,\ket{v_j^\perp,y_{j,k_j}}, \\
&&
j=1,\ldots,m,
\end{eqnarray*}
form an OPB of $\cH$.
\el

\section{Main result}
\label{sec:main}

In this section we present our main result on the multiqubit UPBs in Theorem \ref{thm:the}. In addition to the known results in Sec. \ref{sec:rank5} we shall need two more lemmas.

\begin{lemma} \label{le:suppose}
Suppose there is a UPB $\cU\subset\cH$ of cardinality $2^n-5$ which is orthogonal to five mutually orthogonal states $\ket{a}\ox\ket{\ph_j}$, $j=1,\ldots,5$, with $\ket{a}\in\cH_1$
a unit vector. Then $\cH':=\cH_2\ox\cdots\ox\cH_n$ has a UPB of cardinality $2^{n-1}-5$ which is orthogonal to the states $\ket{\ph_j}$.
\end{lemma}
\begin{proof}
The set $\cV:=\cU\cup\{\ket{a}\ox\ket{\ph_j}:j=1,\ldots,5\}$ is an OPB when we view $\cH$ as a bipartite system 
$\cH=\cH_1\ox\cH'$. By Lemma \ref{le:2xN}, $\cV$ has the form
\begin{eqnarray}
\cV=\{\ket{a_i,\ph_{ij}}:i=1,\ldots,m,~j=1,\ldots,k_i\} 
\notag\\
\cup 
    \{\ket{a_i^\perp,\ps_{ij}}:i=1,\ldots,m,~j=1,\ldots,k_i\},
\end{eqnarray}
where the $\{\ket{a_i},\ket{a_i^\perp}\}$ are $m$ different o.n. 
bases of $\cH_1$. We may assume that $\ket{a}=\ket{a_1}$ and 
$\ket{\ph_j}=\ket{\ph_{1j}}$, $j=1,\ldots,5$. 
As $\braket{a_i}{a_j}\ne0$ for $i\ne j$, the set 
$\cW:=\{\ket{\ph_{ij}}:i=1,\ldots,m,~j=1,\ldots,k_i\}$ 
is an o.n. basis of $\cH'$. It consists of five entangled 
states $\ket{\ph_j}=\ket{\ph_{1j}}$ and $2^{n-1}-5$ product states $\ket{\ph_{ij}}$ for which $\ket{a_i,\ph_{ij}}\in\cU$. 
Since the subspace spanned by the former contains no product vector, the latter vectors form a UPB in $\cH'$ of cardinality 
$2^{n-1}-5$.
\end{proof}

\bl
\label{le:bipartite}
We shall view the $n$-qubit Hilbert space $\cH$ also as a 
bipartite space $\cH=\cH_1\ox\cH'$, where 
$\cH':=\cH_2\ox\cdots\ox\cH_n$.
Let $\r$ be a state of rank five on $\cH$ which is PPTES as 
$n$-partite state and $2\times m$ as a bipartite state. Assume further that $\ker\r$ has an o.n. basis consisting of 
$n$-partite product vectors. 
Then $\r$ is separable as a bipartite state and has length five, say $\r=\sum^5_{j=1}\proj{a_j,\ps_j}$. If $\r$ is a projector then the $\ket{a_j,\psi_j}$ form an o.n. basis of $\cR(\r)$.
\el
\begin{proof}
Since $\rank\r=5$ we must have $n\ge3$. Let 
$\{\ket{b_j,\phi_j}:\;j=1,\ldots,2^n-5\}$ be the o.n. basis of $\ker\r$ mentioned in the lemma. Since the 
$\ket{b_j,\phi_j}\in\ker\r$, we have 
$\bra{b_j,\phi_j}\r\ket{b_j,\phi_j}=0$ which is equivalent to 
$\bra{b^*_j,\phi_j}\r^\G\ket{b^*_j,\phi_j}=0$. As $\r$ is PPT, 
this implies that all $\ket{b^*_j,\phi_j}\in\ker\r^\G$. 
We conclude that $\rank\r^\G\le5$.

Let $\r'=\tr_1\r$, be the state obtained from $\r$ by tracing out the first qubit, and note that $m=\rank\r'$. As $\rank\r=5$, we must have $3\le m\le5$. If $m=5$ then all assertions follow easily from \cite[Corollary 3(a)]{kck00}. Thus we may assume that $m$ is 3 or 4.

Since $\rank\r=5$, the vectors $\ket{\phi_j}$ span a subspace of dimension at least $2^{n-1}-2$. As $m\ge3$, for some $k$ we have 
$\ket{\phi_k}\notin\ker\r'$. Thus we have a decomposition 
$\ket{\phi_k}=\ket{\a}+\ket{\b}$ with $\ket{\a}\in\cR(\r')$, $\ket{\b}\in\ker\r'$. Since $\ket{b_k,\phi_k}\in\ker\r$ and 
$\ket{b_k,\b}\in\cH_1\ox\ker\r'\subseteq\ker\r$, we 
deduce that $\ket{b_k,\a}\in(\cH_1\ox\cR(\r'))\cap\ker\r$. 
By Lemma \ref{le:supposerho} (i) there is a $\lambda>0$ such that 
$\s:=\r-\lambda\proj{b_k,\a}\ge0$ is a PPT state of 
birank $(4,r-1)$ and $\rank\s'=m-1$, where $\s'=\tr_1\s$.
As $m=3$ or $m=4$, by the Peres-Horodecki criterion $\s$ is separable. Hence, $\r$ is separable as a bipartite state. Since $r\le5$, it follows from Lemma \ref{le:supposerho} (ii) that $L(\s)=4$. 

The last assertion follows from the fact that $\rank\r=L(\r)=5$.
\end{proof}

Let us give an example of a state $\r$ satisfying the conditions 
of the above lemma. We take $n=4$ and 
$$
\r=\proj{0,0,0,0}+\proj{1}\ox\s,
$$
where $\s$ is a well-known 3-qubits PPTES of rank four 
constructed from the 3-qubit UPB of size 4 (see \cite{DiV03}).

On the other hand there exist multiqubit PPTES of rank five not of this type. An example is the extremal four-qubit symmetric PPTES with three-rank $(5, 7, 8)$ constructed in \cite{tah2012} at the end of Sec. III. 

Each of the above two PPTESs answers affirmatively a question 
raised in \cite[Sec. 6]{johnston14}, open problem (1). The range of the first PPTES contains the product state $\ket{0,0,0,0}$.
The range of the second PPTES is the 5-qubit symmetric subspace spanned by symmetric product states. So, the kernel of neither of these two PPTESs is spanned by a multiqubit UPB.  

Another interesting example is the so-called X-type multiqubit PPT state of rank five \cite{gs10,hk16}, which may be yet another example of a multiqubit PPTES of rank five, but so far  we have no proof that this is the case.

Now we are in a position to prove our main result.

\begin{theorem}
\label{thm:the}
There are no $n$-qubit UPBs of cardinality $2^n-5$, i.e., 
$2^n-5 \notin\T_n$.	
\end{theorem}
\begin{proof}
Since $\Theta_3=\{4,8\}$ and $\T_4=\{6,7,8,9,10,12,16\}$, the 
assertion is true for $n=3,4$. 
Assume that the assertion fails for some $n>4$ and let $n$ be the smallest such integer. We have to derive a contradiction.

By our assumtion, there exists a UPB $\cU\subset\cH$ of cardinality $2^n-5$. We can write it as
$$
\cU=\{\ket{u_i,\phi_i}:i=1,2,\ldots,2^n-5\},
$$
where $\ket{u_i}\in\cH_1$ are unit vectors and the 
$\ket{\phi_i}$ are product vectors in 
$\cH':=\cH_2\ox\cdots\ox\cH_n$. 

Denote by $\r$ the projector onto the 5-dimensional subspace 
$\cU^\perp$. By Lemma \ref{le:bipartite} we have 
\begin{equation} \label{eq:rho51}
\r=\sum^5_{j=1}\proj{a_j} \ox \proj{\psi_j},	
\end{equation}
where the vectors $\ket{a_j,\psi_j}$ form an o.n. basis of 
$\cU^\perp$ and the $\ket{a_j}\in\cH_1$ are unit vectors. 
If all the $\ket{a_j}$ are equal, then Lemma \ref{le:suppose}
implies that $\cH'$ has a UPB of cardinality $2^{n-1}-5$. 
This contradicts our choice of $n$. We conclude that the 
vectors $\ket{a_j}$ span $\cH_1$. We shall now distinguish 
three cases.

Case 1. $\{\ket{a_j}\}$ contains two o.n. bases of $\cH_1$. 

Without any loss of generality we may assume that 
$\ket{a_2}=\ket{a_1^\perp}$, $\ket{a_4}=\ket{a_3^\perp}$ 
and either $\ket{a_5}\ne\ket{a_j}$ for $j<5$ or 
$\ket{a_5}=\ket{a_4}$. 
Since the $\ket{a_j,\psi_j}$ are pairwise orthogonal, we have $\ket{\ps_3}\perp\ket{\ps_1},\ket{\ps_2}$. By using 
Lemma \ref{le:2xN}, one can verify that 
$\{\ket{\psi_1},\ket{\psi_3}\}^\perp$ has an o.n. basis 
consisting of product vectors. Consequently, 
$\a:=\proj{\ps_1}+\proj{\ps_3}$ is an $(n-1)$-partite PPT state of rank two. By \cite[Lemma 11]{cd13} it is separable of length two. The same is true for $\b:=\proj{\ps_2}+\proj{\ps_3}$, and
so we have
\begin{eqnarray}
\label{eq:ab}
\a &=& \proj{b_2,\ldots,b_n}+\proj{c_2,\ldots,c_n},
\notag\\
\b &=& \proj{d_2,\ldots,d_n}+\proj{e_2,\ldots,e_n},	
\end{eqnarray}
where $\ket{b_2,\ldots,b_n}\perp\ket{c_2,\ldots,c_n}$ and $\ket{d_2,\ldots,d_n}\perp\ket{e_2,\ldots,e_n}$ are unit  vectors. Hence 
\begin{eqnarray}
\label{eq:ps1ps3}
\ket{\ps_1},\ket{\ps_3}\in\lin\{\ket{b_2,\ldots,b_n},\ket{c_2,\ldots,c_n}\},
\notag\\
\ket{\ps_2},\ket{\ps_3}\in\lin\{\ket{d_2,\ldots,d_n},\ket{e_2,\ldots,e_n}\}. 	
\end{eqnarray}

By permuting the last $n-1$ qubits, we may assume that for some 
$m\ge2$ we have $\ket{b_j}\ne\ket{c_j}$ for $2\le j\le m$ and 
$\ket{b_j}=\ket{c_j}$ for $j>m$. 
As no $\ket{\ps_j}$ is a product vector, it follows from 
\eqref{eq:ps1ps3} that $m>2$. 

Suppose that $m>3$. Then Lemma \ref{le:Kruskal} implies that the 
sets $\{\ket{b_2,\ldots,b_n},\ket{c_2,\ldots,c_n}\}$ and 
$\{\ket{d_2,\ldots,d_n},\ket{e_2,\ldots,e_n}\}$ are the same 
(up to phase factors). It follows that the vectors 
$\ket{\ps_1}$, $\ket{\ps_2}$, $\ket{\ps_3}$ belong to the same 2-dimensional subspace, and since 
$\ket{\ps_3}\perp\ket{\ps_1},\ket{\ps_2}$ we deduce that
$\ket{\ps_1}=\ket{\ps_2}$. Since $\ket{a_1,\ps_1}$ and 
$\ket{a_2,\ps_2}=\ket{a_1^\perp,\ps_1}$ belong to $\cR(\r)$, so does $\ket{a_3,\ps_1}$. Since also $\ket{a_3,\ps_3}\in\cR(\r)$, we deduce that $\ket{a_3,b_2,\ldots,b_n}\in\cR(\r)$. This 
contradicts the assumption that $\cU$ is a UPB. 

Hence, we must have $m=3$ and so
\begin{equation}
\label{eq:b4}
\ket{b_4,\ldots,b_n}=\ket{c_4,\ldots,c_n}=
\ket{d_4,\ldots,d_n}=\ket{e_4,\ldots,e_n}.	
\end{equation}
As mentioned earlier, we have to consider two possibilities for 
$\ket{\ps_5}$. 

The first one is that $\ket{a_5}\ne\ket{a_j}$ for $j<5$. By using Lemma \ref{le:2xN}, one can verify that the subspace $\{\ket{\ps_1},\ket{\ps_3},\ket{\ps_5}\}^\perp$ of $\cH'$ has an o.n. basis consisting of product vectors. As the vectors $\ket{a_j,\psi_j}$ are mutually orthogonal, we have 
$\ket{\ps_5}\perp \ket{\ps_1},\ket{\ps_3}$. Hence, by adjoining the product vectors $\ket{b_2,\ldots,b_n}$ and 
$\ket{c_2,\ldots,c_n}\}$ to the above mentioned o.n. basis, we obtain an o.n. basis for the hyperplane of $\cH'$ orthogonal to 
$\ket{\ps_5}$. As $\ket{\ps_5}$ is not a product vector and $\cH'$ has no UPBs of cardinality $2^{n-1}-1$, we have a contradiction.

The second possibility is that $\ket{a_5}=\ket{a_4}$. As the vectors $\ket{a_j,\psi_j}$ are mutually orthogonal, the same is true for the vectors $\ket{\ps_1}$, $\ket{\ps_4}$ and 
$\ket{\ps_5}$.
By using Lemma \ref{le:2xN}, one can verify that the subspace $\{\ket{\ps_1},\ket{\ps_4},\ket{\ps_5}\}^\perp$ of $\cH'$ has an o.n. basis consisting of product vectors. 
Since $2^{n-1}-k\notin\T_{n-1}$ for $k\in\{1,2,3\}$, the 
subspace spanned by $\ket{\ps_1}$, $\ket{\ps_4}$ and 
$\ket{\ps_5}$ has an o.n. basis consisting of product vectors, 
say $\ket{p}=\ket{p_2,\ldots,p_n}$, 
$\ket{q}=\ket{q_2,\ldots,q_n}$ and 
$\ket{r}=\ket{r_2,\ldots,r_n}$. It follows that there exists an order-3 unitary matrix $[u_{ij}]$ such that 
\begin{eqnarray}
\label{eq:ps1ps4}
\ket{\ps_1}=u_{11}\ket{p}+u_{12}\ket{q}+u_{13}\ket{r}, \notag\\
\ket{\ps_4}=u_{21}\ket{p}+u_{22}\ket{q}+u_{23}\ket{r}, \notag\\
\ket{\ps_5}=u_{31}\ket{p}+u_{32}\ket{q}+u_{33}\ket{r}.
\end{eqnarray}
By the argument we used above to prove that the $\ket{a_j}$ are 
not all equal, we can show that there is no $j>3$ such that 
$\ket{p_j}=\ket{q_j}=\ket{r_j}=\ket{b_j}$. 

If $u_{11}=0$ then we have $u_{22}u_{33}-u_{23}u_{32}=0$ and by
taking a suitable linear combination of $\ket{a_4,\ps_4}$ 
and $\ket{a_4,\ps_5}$ we obtain that the product vector 
$\ket{a_4,p}\in\cR(\r)$, i.e. we have a contradiction. 
Thus $u_{11}\ne0$, and similarly $u_{12}\ne0$ and $u_{13}\ne0$. 
From the equations \eqref{eq:ps1ps3} and \eqref{eq:b4} we obtain that
\begin{eqnarray}
\label{eq:xy}	
\ket{\ps_1}
&=&
(\xi\ket{b_2,b_3} + \eta\ket{c_2,c_3}) \ox \ket{b_4,\ldots,b_n}
\notag\\
&=&
u_{11}\ket{p_2,\ldots,p_n}
+u_{12}\ket{q_2,\ldots,q_n}
+u_{13}\ket{r_2,\ldots,r_n},
\notag\\
\end{eqnarray}
where $\xi\eta\ne0$ and $\ket{b_j}\ne\ket{c_j}$ for $j=2,3$. 
Since we have shown that for $j>3$ at least one of the three unit vectors $\ket{p_j}$, $\ket{q_j}$, $\ket{r_j}$ is not equal to $\ket{b_j}$, the equation \eqref{eq:xy} implies that at most one of the same three vectors can be equal to $\ket{b_j}$.

We claim that $\ket{p_j}\ne\ket{b_j}$ for $j>3$. We shall prove it by contradiction. Assume that, say, $\ket{p_n}=\ket{b_n}$. Then the equation \eqref{eq:xy} implies that 
$\ket{q_2,\ldots,q_{n-1}}=\ket{r_2,\ldots,r_{n-1}}$ and so 
\begin{eqnarray}
\label{eq:prosta}	
\ket{\ps_1}&=&
u_{11}\ket{p_2,\ldots,p_n}
+\ket{q_2,\ldots,q_{n-1}} \ox (u_{12}\ket{q_n}+u_{13}\ket{r_n}).
\notag\\
\end{eqnarray}
From \eqref{eq:xy} we see that 
$\braket{b_4,\ldots,b_n}{\psi_1}=
\xi\ket{b_2,b_3} + \eta\ket{c_2,c_3}$ has Schmidt rank 2. 
From \eqref{eq:prosta} we deduce that also 
$\ket{p_2,p_3} + \ket{q_2,q_3}$ has Schmidt rank 2. As 
$\braket{b_{n-1}^\perp}{\psi_1}=0$, from the same equation 
we obtain that
\begin{eqnarray}
u_{11}\braket{b_{n-1}^\perp}{p_{n-1}}\ket{p_2,\ldots,p_{n-2},p_n}
+\notag\\
\braket{b_{n-1}^\perp}{q_{n-1}} \ket{q_2,\ldots,q_{n-2}}
\ox (u_{12}\ket{q_n}+u_{13}\ket{r_n})=0.
\notag
\end{eqnarray}
This equation implies that 
$\braket{b_{n-1}^\perp}{p_{n-1}}=
 \braket{b_{n-1}^\perp}{q_{n-1}}=0$ which gives the 
contradiction: $\ket{p_{n-1}}=\ket{q_{n-1}}=\ket{b_{n-1}}$.
Thus our claim is proved.

Let $\ket{\theta}=\ket{\psi_1}-u_{11}\ket{p}$. By switching 
the tensor factors $\cH_3$ and $\cH_4$, $\ket{\theta}$ is mapped 
to 
\begin{eqnarray}
\xi\ket{b_2,b_4}\ox\ket{b_3,b_5,\ldots,b_n}+ 
 \eta\ket{c_2,b_4}\ox\ket{c_3,b_5,\ldots,b_n}
\notag\\
- 
 u_{11}\ket{p_2,p_4}\ox\ket{p_3,p_5,\ldots,p_n}.
\notag
\end{eqnarray}
The three first tensor factors, namely $\ket{b_2,b_4}$, 
$\ket{c_2,b_4}$ and $\ket{p_2,p_4}$, are linearly independent 
and the same holds true for the second tensor factors. 
It follows that $\ket{\theta}$ considered as bipartite tensor 
has tensor rank 3. Now the equation \eqref{eq:xy} gives a contradiction.

Case 2. $\{\ket{a_j}\}$ contains only one o.n. basis of $\cH_1$.

Say, $\ket{a_2}=\ket{a_1^\perp}$. We have 
$\r=\proj{a_1}\ox\a+\proj{a_1^\perp}\ox\b+\g$, where 
$\a$ is a sum of $p$ terms $\proj{\psi_i}$, 
$\b$ is a sum of $q$ terms $\proj{\psi_j}$, and 
$\g$ is a sum of $5-p-q$ terms $\proj{a_l,\psi_l}$. 
We may also assume that $p\le q$. 
By using Lemma \ref{le:2xN}, one can verify that the 
orthogonal complement in $\cH'$ of the set of the $\ket{\psi_i}$ 
which appear in $\a$ has an o.n. basis consisting of product vectors. Hence, $\ket{\a}$ is a PPT state of rank $p\le 2$. By Lemma \ref{le:theppt} $\a$ is separable. It follows that $\cR(\r)$ contains a product vector 
and we have a contradiction.

Case 3. $\{\ket{a_j}\}$ contains no o.n. basis of $\cH_1$.

We may assume that $\ket{a_i}=\ket{a_1}$ for $i\le p$ and $\ket{a_j}\ne\ket{a_1}$ for $j>p$. Let $V\subseteq\cH'$ be 
the subspace spanned by the $\ket{\psi_i}$ for $i\le p$. 
By the same argument as in case 2, the subspace 
$V^\perp\subseteq\cH'$ has an o.n. basis consisting of product vectors. Note that $p\notin\T_{n-1}$. This is clear when $p<5$ and it is true for $p=5$ by our choice of $n$. 
Consequently $V$ contains a product vector, say $\ket{\phi}$.
Then $\ket{a_1,\phi}$ is a product vector in $\cR(\r)$ and we have a contradiction. 

This completes the proof.	
\end{proof}

%

Let $\r$ be an $n$-qubit PPTES of rank five. Since any 3-qubit subspace of dimension five contains a product vector, we must have $n>3$. Theorem \ref{thm:the} shows that $\ker\r$ is not spanned by a UPB. However $\cR(\r)$ may contain a product vector. Indeed we have shown that this is the case for the PPTESs mentioned below Lemma \ref{le:bipartite}. 
If such $\r$ exists for $n=4$ then $\r\ox\proj{a}$, with $\ket{a}$ a product vector, would be also a multiqubit PPTES of rank five whose range contains no product vectors. Intuitively we believe that

\begin{conjecture}
\label{cj:ppt}
The range of any multiqubit PPTES of rank five contains a nonzero product vector. 
\end{conjecture}
Note that this conjecture is stronger than Theorem \ref{thm:the}.

\section{Conclusions}
\label{sec:con}

The problem of constructing UPBs in multipartite quantum systems, which is more than 15 years old, is still of interest
due to its role in various applications such as those mentioned in the Introduction. Some important advances in the case of multiqubit systems have been made recently regarding the cardinalities of UPBs in such systems. For instance the minimal size of the $n$-qubit UPB is known for all $n$. Recall that $\T_n$ denotes the set of sizes of the UPBs of the $n$-qubit system. In this paper we have proved, see Theorem \ref{thm:the},  that $2^n-5\notin\T_n$ for any $n$. In particular, 
$27\notin\Theta_5$, $59\notin\Theta_6$, and $123\notin\Theta_7$. 
This gives a partial answer to a question raised in \cite{johnston14}. We propose a conjecture about multiqubit entangled PPT states of rank 5, Conjecture \ref{cj:ppt}. 

Let us also point out that the improper UPBs (known as OPBs) of $n$-qubit systems, i.e. those of cardinality $2^n$, can be studied by using special combinatorial matrices of size 
$2^n\times n$ as in our paper \cite{cd16}. With some minor modifications, the same combinatorial technique is applicable 
to the study of proper UPBs.

\section*{Acknowledgements}

LC was supported by Beijing Natural Science Foundation (4173076), the NNSF of China (Grant No. 11501024), and the Fundamental Research Funds for the Central Universities (Grant Nos. KG12001101, ZG216S1760 and ZG226S17J6). 
The second author was supported in part by the National Sciences and Engineering Research Council (NSERC) of Canada Discovery 
Grant 5285.

\end{document}